\pdfoutput=1

\pdfoutput=1
\documentclass[aps,twocolumn]{revtex4-2}
\usepackage{graphicx} 
\usepackage{amsthm} 
\usepackage{amsmath}
\usepackage{mathtools}
\usepackage{bbold}  
\usepackage{amsfonts} 
\usepackage{physics} 
\usepackage{graphicx}
\usepackage[caption=false]{subfig}
\usepackage{xcolor} 
\usepackage{tikz}
\usepackage{tikz-cd} 
\usepackage{adjustbox} 
\usepackage{centernot}
\usepackage[T1]{fontenc}
\usepackage[normalem]{ulem} 

\usetikzlibrary{decorations.pathreplacing}
\usetikzlibrary{decorations.pathmorphing}

\DeclareFontFamily{U}{mathb}{\hyphenchar\font45} 
\DeclareFontShape{U}{mathb}{m}{n}{
<-6> mathb5 <6-7> mathb6 <7-8> mathb7
<8-9> mathb8 <9-10> mathb9
<10-12> mathb10 <12-> mathb12
}{}
\DeclareSymbolFont{mathb}{U}{mathb}{m}{n}
\DeclareMathSymbol{\llcurly}{\mathrel}{mathb}{"CE}
\DeclareMathSymbol{\ggcurly}{\mathrel}{mathb}{"CF}

\newcommand{\thermomaj}{\succ_{\rm th}}

\newcommand{\cthermomaj}{\ggcurly_{\rm th}}

\tikzset{
    level/.style = {
        ultra thick,
        black,
    },
    connect/.style = {
        dashed,
        red
    },
    notice/.style = {
        draw,
        rectangle callout,
        callout relative pointer={#1}
    },
    label/.style = {
        text width=1cm
    }
}

\theoremstyle{definition}
\newtheorem{definition}{Definition}
\newtheorem{theorem}{Theorem}
\newtheorem{lemma}{Lemma}

\begin{document}

\title{Capacity of non-Markovianity to boost the efficiency of molecular switches}
\author{Giovanni Spaventa}
\email{giovanni.spaventa@uni-ulm.de}
\affiliation{Institute of Theoretical Physics \& IQST, Ulm University, Albert-Einstein-Allee 11 89081, Ulm, Germany}
\author{Susana F. Huelga}
\email{susana.huelga@uni-ulm.de}
\affiliation{Institute of Theoretical Physics \& IQST, Ulm University, Albert-Einstein-Allee 11 89081, Ulm, Germany}
\author{Martin B. Plenio}
\email{martin.plenio@uni-ulm.de}
\affiliation{Institute of Theoretical Physics \& IQST, Ulm University, Albert-Einstein-Allee 11 89081, Ulm, Germany}
\date{\today}

\begin{abstract}
Quantum resource theory formulations of thermodynamics offer a versatile tool for the study of fundamental limitations to the efficiency of physical processes, independently of the microscopic details governing their dynamics. Despite the ubiquitous presence of non-Markovian dynamics in open quantum systems at the nanoscale, rigorous proofs of their beneficial effects on the efficiency of quantum dynamical processes at the molecular level have not been reported yet.
Here we combine the quantum resource theory of athermality with concepts from the theory of divisibility classes of quantum channels, to prove that memory effects can increase the efficiency of photoisomerization to levels that are not achievable under a purely thermal Markovian (i.e. memoryless) evolution. This provides rigorous evidence that memory effects can provide a resource in quantum thermodynamics at the nanoscale.
\end{abstract}

\maketitle

{\bf\emph{Introduction --}} Photoinduced switching in molecular systems is at the basis of many fundamental processes in living organisms. Relevant examples include light activated processes in the primary steps of photosynthesis in plants, algae and bacteria \cite{croce2018light} and the photoisomerization of retinal in the cells of human's retina \cite{schulten1978biomagnetic}. Furthermore, molecular switches can be artificially engineered and employed in technological applications, such as the storage of solar energy, nanorobotics and optical data storage \cite{dattler2019design}.
This class of reactions involve sub-picosecond time scales and are characterized by remarkably high quantum yields and strong specificity \cite{nogly2018retinal,seidner1994microscopic,seidner1995nonperturbative,hahn2000quantum,hahn2002ultrafast}. Here, the external irradiation can
be thought of as a high-temperature reservoir, with the vibrations of the scaffold of the pigment-protein complex playing the role 
of a low-temperature reservoir for excess energy. This framework suggests a \textit{thermal engine} representation of the process and a thermodynamic analysis 
in which to determine stringent measures of efficiency as well as fundamental limits thereof. Recent work, for instance, compared the output power of model pigment networks subject to purely thermal relaxation with models that include a coherent coupling of the network to the protein modes \cite{killoran2015enhancing}. This type of analysis allows to draw quantitative conclusions concerning the properties that enhance efficiency, exciton-vibrational coherence in that case, but at the same time is very much constrained by the specific dynamical details of the model. Therefore, it would be desirable to perform such an analysis at a higher level of abstraction, in the same spirit as the study of classical thermal machines using the laws of macroscopic thermodynamics. This is precisely what resource theories \cite{coecke2016mathematical,chitambar2019quantum}, developed in the context of quantum information, can facilitate. In particular, the resource theory of thermodynamics at the nanoscale \cite{ruch1976principle,ruch1978mixing,janzing2000thermodynamic,horodecki2013fundamental,goold2016role,lostaglio2019introductory,ng2018resource} has been applied very recently to determine the optimal photoisomerization yield of a molecular switch \cite{halpern2020fundamental}. Here, we use the same model system to first obtain an analytical bound that encompasses previous numerical results and then focus on extending the resource theoretic formalism from general thermal operations, characteristic of an input-output formalism at specific moments in time, to continuous processes which allow us to connect the concept of (non)-Markovianity with thermal operations. Building on this, we then proceed to demonstrate the existence of a finite gap between the optimal yields achievable under general thermal operations and that achievable under Markovian operations alone.
This allows us to rigorously quantify the yield gain that is due to non-Markovian effects in the system-environment interaction \cite{rivas2014quantum,breuer2016colloquium,devega2017dynamics,li2018concepts,milz2020quantum}. This result is of particular importance for the dynamics of bio-molecular complexes for which the presence of non-Markovianity, due to the interplay of the electronic quantum dynamics with highly structured environmental spectral densities and long-lived vibrational modes, is well-established \cite{huelga2013vibrations}.

{\bf\emph{Thermodynamics as a resource theory --}} Quantum resource theories (QRTs) provide a theoretical framework in which a set of states/operations are considered free, and any state/operation outside the free set can then become a resource, in the sense of facilitating a task inaccessible to the free set. Non-free states are therefore deployed at a cost, but in exchange they assist processes that would be otherwise impossible or only attainable with a smaller fidelity. Historically, the first example of a resource theory is the theory of bipartite entanglement \cite{plenio2014introduction,horodecki2009quantum}, where the restriction to local operations and classical communication (LOCC) promotes entanglement to a resource, while separable states remain freely accessible. Analogously, quantum thermodynamics can be formulated as a resource theory in which both free states and operations are \emph{thermal} \cite{ruch1976principle,ruch1978mixing,janzing2000thermodynamic,horodecki2013fundamental,goold2016role,lostaglio2019introductory,ng2018resource}.  Specifically, the resource theory of \emph{athermality} is constructed as follows. Given a system $S$ with Hamiltonian $H_S$, the following three elementary operations are allowed: (i) The system can be brought into contact with a thermal bath $B$, that is, we can freely deploy Gibbs states $\tau = e^{-\beta H_B}/Z$ at inverse temperature $\beta$. (ii) We can perform any global unitary transformation $U$ on $S+B$, as long as it is energy preserving, i.e., $ \big[ U,H_S+H_B \big]=0$ . (iii) We are allowed to trace out subsystems, and in particular the entire bath $B$.
As a result, the action of \emph{thermal operations} (TO) on a density operator $\rho_S$ is then defined as
\begin{equation} \rho_S \,\xrightarrow{\,\,\rm{TO}\,\,}\, \Tr_{B}\, \big[ U\,\rho_S\otimes\tau\, U^\dagger \big]\,. \end{equation}
Note that thermal operations preserve the Gibbs state of the system $S$, and furthermore they obey \textit{time-translation covariance} (also  called \textit{phase-covariance} or $U(1)$-covariance), i.e. they commute with the free unitary evolution of the system: $\mathcal{T}\circ\mathcal{U}_t=\mathcal{U}_t\circ\mathcal{T}$ for any $\mathcal{T}\in\mathsf{TO}$.

The action of thermal operations on quasiclassical states (states that are diagonal in the energy eigenbasis) can be fully characterised as that of matrices mapping the population vector of an initial state to the population vector of a final state, given the phase covariance of the operation. More concretely, thermal operations act on population vectors as Gibbs-stochastic matrices, i.e., stochastic matrices that preserve the Gibbs state. We will denote by $\mathsf{GS}_n$ the semigroup (or, more precisely, the monoid \footnote{For fixed $H$, the set $\mathsf{GS}_n(\beta H)$ is closed under matrix multiplication, therefore it is a semigroup. Furthermore, the identity matrix trivially belongs to this set, making $\mathsf{GS}_n(\beta H)$ a semigroup with an identity element, i.e. a monoid.}) of $n\times n$ stochastic matrices that preserve the diagonal of the Gibbs-state $e^{-\beta H}/Z$. Then, given two quasiclassical states $\rho$ and $\sigma$, the first can be mapped into the second by thermal operations ($\rho\, \xrightarrow{\mathsf{TO}}\, \sigma$) if and only if there exists a Gibbs-stochastic matrix mapping the diagonal $p$ of $\rho$ to the diagonal $q$ of $\sigma$:
\begin{equation} \rho\, \xrightarrow{\mathsf{TO}}\, \sigma\,\iff\, \exists\, G\in \mathsf{GS}_n\, \text{  s.t.  } q = Gp\,.  \end{equation}
The problem of assessing when two particular states are connected by a Gibbs-stochastic matrix is then solved by \textit{thermomajorization} \cite{ruch1976principle,ruch1978mixing,horodecki2013fundamental}, in a way analogous to how majorization characterizes state convertibility under LOCC in the resource theory of entanglement. 
In particular, one associates to a density matrix $\rho$ a curve $f_\rho(x)$ (called \textit{thermomajorization curve}) and, given two density matrices $\rho$ and $\sigma$, it is said that $\rho$ thermomajorizes $\sigma$, i.e. $\rho\thermomaj \sigma$, if
$ f_\rho (x) \geq f_\sigma (x) \,, \forall x$.
Then, given two quasiclassical states $\rho$ and $\sigma$:
\begin{equation} \rho\, \xrightarrow{\mathsf{TO}}\, \sigma\,\iff\, \rho \thermomaj \sigma\,.  \end{equation}

{\bf\emph{Photoisomerization --}} We now consider the problem of modeling photoisomerization in the resource theory of athermality.  We will use the same model system as in \cite{halpern2020fundamental}, where an angular coordinate $\varphi$ between two heavy chemical groups parametrizes the relative rotation of two molecular components around a double bond. Fig.\ref{fig:energy_landscape} displays a typical energy landscape for these systems, where the two eigenvalues $\mathcal{E}_{0,1}(\varphi)$ can be obtained from a class of Hamiltonians commonly used in the study of photoisomerization (see \cite{seidner1994microscopic,seidner1995nonperturbative,hahn2000quantum,hahn2002ultrafast}). It is important to stress, however, that the results presented in this work are completely independent of the specific form chosen for the Hamiltonian (and therefore the energy landscape of the system), as the only parameters that are crucial to our analysis are the energies of the levels that are occupied at the initial and final times and their corresponding populations.
The ground state energies for $\phi=0$ and $\phi=\pi$ satisfy $\mathcal{E}_{0}(\varphi=0)\le\mathcal{E}_{0}(\varphi=\pi)$
in our analysis, and are separated by an energy barrier.
The molecule begins in a thermal state of configuration $\varphi=0$. Then, following photo-excitation, the molecule can isomerize (or switch configuration) while relaxing in contact with its environment, for example via a dissipative Landau-Zener transition (the results are however independent of the intermediate mechanisms governing the process, which could for instance involve a conical intersection \cite{polli2010conical}). The probability of switching configuration during relaxation is called \textit{photoisomerization yield}.

\begin{figure}
\centering\includegraphics[width=0.48\textwidth]{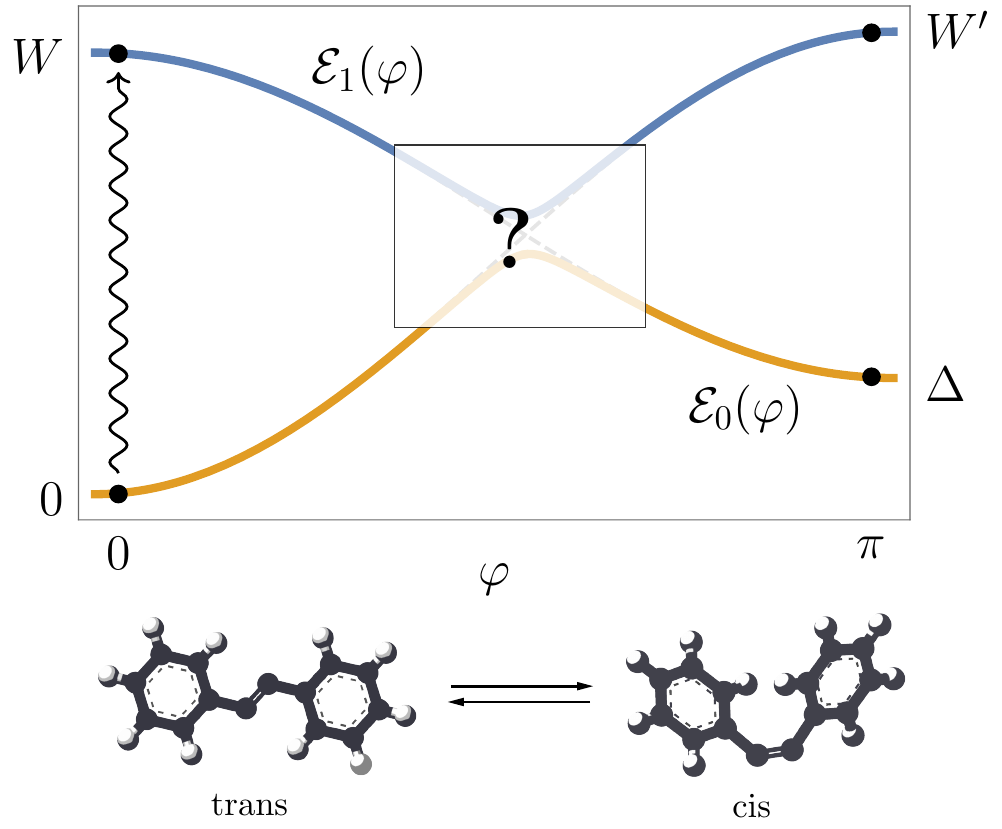}
\caption{Energy landscape for a typical photoisomer. The system starts in the electronic ground state at $\varphi=0$, and it is photoexcited (wavy arrow) by a light source. It can then relax to the cis ground state $\mathcal{E}(\varphi=\pi)$, while in contact with its environment. Our results are independent of the actual intermediate dynamics of the process. The advantage of a resource theoretic approach, is precisely that of providing general bounds, that are valid regardless of the complicated microscopic details governing the evolution of the system. The four dots represent the states that are considered in the four-levels model of \cite{halpern2020fundamental}.}
    \label{fig:energy_landscape}
\end{figure}

\noindent The molecule starts in thermal equilibrium with its environment at an inverse temperature $\beta$, at an angular configuration $\varphi=0$. The first step is that of \textit{photoexcitation}, during which a light source excites the electronic state of the molecule at the fixed angular coordinate $\varphi =0$. As a consequence, the molecule ends up in a new electronic state $\rho_i$. Two molecular components can then rotate relative to each other while relaxing in contact with the environment, and the molecule ends up in a state $\rho_f$ which will have some weight on the \textit{cis} configuration $\varphi=\pi$. Finally, the molecule thermalizes again.\\ The \textit{photoisomerization yield} $\gamma$ is then defined as the weight of the post-rotation state on the cis electronic ground state, that is $\gamma = \bra{\mathcal{E}_0(\pi)} \rho_f \ket{\mathcal{E}_0(\pi)}$.\\
The rotation step can then be thought as a mapping $\rho_i\mapsto\rho_f$ between the post-excitation and the post-rotation electronic states. In the absence of external sources of work and/or coherence, and under the assumption that the total energy of system and environment is conserved, the mapping above can be described as a thermal operation. Indeed, thermal operations are precisely those operations that do not require the consumption of any initial source of work and/or coherence. In real world scenarios, the actual process could depart from thermal operations, for example as a consequence of total energy of system and bath not being exactly conserved, or the presence of an external clock, which would break time-translation covariance. The results presented in this work could then be used as a witness of athermality. In other words, a violation of the bounds presented here, would rigorously prove that the mapping $\rho_i\mapsto\rho_f$ is not a thermal operation, which in turn would mean that the specific environment under analysis does not merely act as a passive thermal reservoir, but as a battery (a source of work) and/or as a clock (a source of asymmetry).

Photoexcitation changes the electronic state of the photoisomer to a state $\rho_i$ which in general will not be diagonal in the energy eigenbasis, i.e. $[\rho,H]\neq0$. One would then naturally expect coherence in the energy eigenbasis to affect the yield. However, the time-translation covariance of thermal operations guarantees that populations evolve independently from coherences. Therefore, if we are only interested in the yield, we can disregard coherences completely. In fact, the  yield is nothing more than the population of level $\ket{\mathcal{E}_0(\pi)}$, and as such, it cannot possibly be affected by the presence of coherence in the energy eigenbasis, allowing us to focus on quasiclassical states only. With this in mind, the transformation $\rho_i\xrightarrow{}\rho_f$ is possible via thermal operations if and only if the initial state thermomajorizes the final state. This means that not all values of the yield are allowed by the constraints of thermal operations, and that we could find an upper bound to the yield by solving for the largest value of $\gamma$ such that $\rho_i \thermomaj \rho_f$.

A numerical solution for this problem is provided in \cite{halpern2020fundamental}, where the authors focus on states with $\varphi=0$ and $\varphi=\pi$, effectively reducing the photoisomer to a four-levels system $\{\ket{\mathcal{E}_0(\varphi)},\ket{\mathcal{E}_1(\varphi)}\}_{\varphi=0,\pi}$.\\ This simplification is made possible by the fact that the thermomajorization conditions only involve the initial and final states, together with the assumption on the initial state of the molecule having negligible overlap with $\varphi\neq0$ configurations. A more sophisticated model, that takes into account the finite width of the initial distribution around $\varphi=0$, would require the addition of levels near the stable configurations, and represents an interesting extension to the analysis.

{\bf\emph{Optimal yield for a single photoisomer --}}
We will now obtain an analytical solution to the maximum allowed yield under general thermal operations.
For that, it suffices to model the photoisomer as a three-levels system $\{\ket{0},\ket{\Delta},\ket{W}\}$ (see Appendix \ref{appendix_GS4}). These three levels have energies $0,\Delta,W$, respectively, and they arise from the full spectrum $\mathcal{E}_{0,1}(\varphi)$ by focusing on  $\varphi=0,\pi$ only, and by neglecting the uppermost energy level $\ket{\mathcal{E}_1(\pi)}$. Given a population vector $(p_0,p_\Delta,p_W)$, the photoisomerization yield is then $\gamma(\rho) := \bra{\Delta} \rho \ket{\Delta}=p_\Delta$.

\begin{figure}[t]
\centering
\includegraphics[]{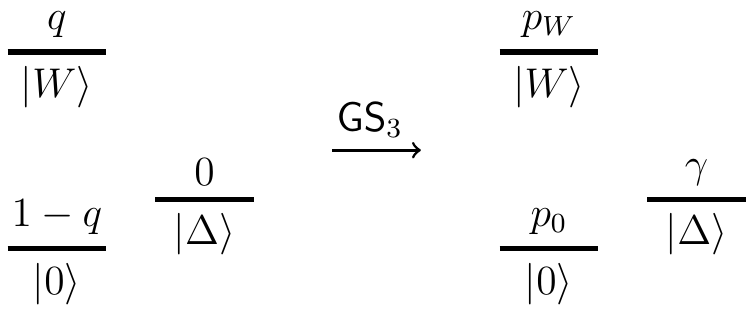}


\caption{Photoisomerization as a Gibbs-stochastic mapping between three-levels systems. The three states $\{\ket{0},\ket{\Delta},\ket{W}\}$ arise from focusing on the two angles $\varphi=0,\pi$ only, and ignoring the uppermost level $\mathcal{E}_1(\pi)$. The populations in the initial state are $(1-q,0,q)$, while the final state has the photoisomerization yield $\gamma$ as the new population for the level $\Delta$.}
\label{fig:three_levels}
\end{figure}
Let us consider an initial state of the system parametrized by the \textit{photoexcitation factor} $q\in[0,1]$:
\begin{equation}
    \rho_i \longmapsto\, \big(1-q,0,q\big)\,,
    \label{eq_state_param}
\end{equation}
our aim is to constrain the possible evolution in the populations of these three states under the action of thermal operations.
The situation is depicted in Fig.\ref{fig:three_levels}, where a Gibbs-stochastic matrix transforms the initial state $\rho_i$ of the photoisomer to a final state with yield $\gamma$.

If we denote by $\gamma(G)$ the yield produced by the matrix $G$, the optimal photoisomerization yield achievable under thermal operations is
\begin{equation}
    \gamma^*:=\sup_{G\in \mathsf{GS}_3} \gamma(G) \,.
\end{equation}
From now on, we will measure energies in units of $\beta^{-1}$, so as to simplify the notation.
By considering a $3\times 3$ stochastic matrix, and by imposing that the Gibbs state is preserved, we find that the general form of a matrix $G\in \mathsf{GS}_3$ is
\begin{equation}
    \begin{pmatrix} 1-g_1e^{-\Delta}-g_2e^{-W} & g_1 & g_2 \\ (1-g_3)e^{-\Delta} - g_4 e^{-W} & g_3 & g_4 \\ (g_2+g_4)e^{-W}{-}(1{-}g_1{-}g_3)e^{-\Delta} & 1{-}g_1{-}g_3 & 1{-}g_2{-}g_4 \end{pmatrix}
    \label{eq:generic_GS3}
\end{equation}
where $g_1,g_2,g_3,g_4\in[0,1]$. All these parameters have to obey a set of constraints ensuring that every element in the matrix above is non-negative.

By computing the action of such a matrix on the initial state $\rho_i$, we obtain the following expression for the photoisomerization yield $\gamma(G)$:
\begin{equation}
    \gamma(g_3,g_4) = \left[(1-q)e^{-\Delta}\right](1-g_3) + \left[ (1+e^{-W})q-e^{-W}\right]g_4\,.
\end{equation}
This is a linear function of $g_3$ and $g_4$, and it can be easily maximized, producing the \textbf{optimal \mbox{photoisomerization} yield}:
\begin{equation}
    \gamma^* = \begin{cases} q+(1-q)\big( e^{-\Delta} - e^{-W} \big) \quad &\text{if}\,\,q\geq\tilde{q}\,, \\
    (1-q)e^{-\Delta} \quad &\text{if}\,\,q<\tilde{q}\,,\end{cases}
\end{equation}
where we have defined $\tilde{q}=1/(1+e^W)$.\\

{\bf\emph{Markovianity and thermal operations --}}
Building on the previous section, we proceed with the demonstration that the optimal yield $\gamma^*$ cannot be achieved under the additional restriction of thermal Markovianity. This result demonstrates that, in the absence of additional resources, some non-Markovianity (and therefore some memory) is fundamentally required to optimize the process of photoisomerization.

\noindent The first problem we face is that of relating the property of Markovianity with the framework of thermal operations. In fact, since the latter are defined in a way that does not explicitly involve any form of dynamics, there is no immediate way of telling if a thermal operation represents interactions with a Markovian bath or not.

Consider a thermal operation $\mathcal{T}\in \mathsf{TO}$, mapping a state $\rho_i$ to a state $\rho_f$, i.e. $\rho_i \xrightarrow{\mathcal{T}}\rho_f$. If such a transformation has to be physically meaningful, in the sense that it originates from a real physical evolution of a system, there must exist an underlying continuous dynamics, represented by a completely positive trace preserving (CPTP) map $\mathcal{E}_{(t,0)}$, which reproduces $\mathcal{T}$ for some value $t$, in other words:
\begin{equation}
    \begin{split}
       \rho_i \xrightarrow{\mathcal{T}}\rho_f  \iff  \exists\,\, \mathcal{E}_{(t,0)}\,\,\text{CPTP}\,\,\text{s.t.}\,\, \rho_f = \mathcal{E}_{(t,0)}(\rho_i)\,.
\end{split}
\label{eq_markovTO}
\end{equation}
 Given a thermal operation $\mathcal{T}$, we want to characterize whether it possible to find an underlying continuous Markovian map of which $\mathcal{T}$ is a particular snapshot.  If not, we will say that $\mathcal{T}$ is inherently non-Markovian, in the sense that there are no Markovian maps that could possibly reproduce the effects of $\mathcal{T}$ on the system.\\
Here we will adopt a definition of Markovianity in terms of the CP-divisibility of the map \cite{rivas2010entanglement}. Namely, $\mathcal{E}_{(t,0)}$ will be called Markovian if it is CP-divisible, i.e. if for all $0\leq s \leq t$ it can be decomposed as
\begin{equation}
    \mathcal{E}_{(t,0)} = \mathcal{E}_{(t,s)}\circ\mathcal{E}_{(s,0)}\,,
\end{equation}
with $\mathcal{E}_{(t,s)}$ a completely positive map \cite{wolf2008dividing,rivas2010entanglement,chruscinski2011measures,buscemi2016equivalence}. 
\begin{definition}[Embeddable thermal operations]
A thermal operation $\mathcal{T}\in \mathsf{TO}$ is said to be \textit{time-independent Markovian} (or \textit{embeddable}) if there exists a Lindblad generator $\mathcal{L}$ such that $\mathcal{E}_{(s,0)}=\exp\left(\mathcal{L}s\right)$ defines a thermal operation $\forall\,s$ and $\mathcal{T}=\mathcal{E}_{(t,0)}$. The set of embeddable thermal operations on $n$-dimensional systems is denoted by $\mathsf{ETO}_n$.
\label{def_ETO}
\end{definition}
\begin{definition}[Markovian thermal operations]
A thermal operation $\mathcal{T}\in \mathsf{TO}$ is said to be \textit{Markovian} (or \textit{memoryless}) if there exists a continuous family of Lindblad generators $\{\mathcal{L}(s)\mid 0\leq s\leq t\}$ such that $\forall\,s>r \in[0,t]$ the time-ordered exponential $\mathcal{E}_{(s,r)}=\mathsf{T}\exp\left( \int_r^s \mathcal{L}(\tau)d\tau \right)$ defines a thermal operation  and $\mathcal{T}=\mathcal{E}_{(t,0)}$. The set of Markovian thermal operations on $n$-dimensional systems is denoted by $\mathsf{MTO}_n$.
\label{def_MTO}
\end{definition}

\noindent These definitions capture the difference between maps that can be generated by time-independent Lindblad generators, and the more general CP-divisible maps. In the case of a thermal operation $\mathcal{T}$ being Markovian, the memoryless property amounts to the possibility of decomposing $\mathcal{T}$ into products of infinitesimal thermal operations $\mathcal{T}(t+\epsilon,t)$. This is consistent with saying that the bath is restored to thermal equilibrium at every step $\epsilon$, in a continuous fashion, therefore effectively ruling out any memory effect. Time-translation covariance holds locally in time with this definition, and the evolution of populations and coherences is therefore decoupled.\\
It is important to stress that Markovian operations, according to Def. \ref{def_MTO}, are generally not embeddable. However, they can be approximated arbitrarily well by sequences of embeddable operations, each with their own generator, $\mathcal{T}_i=\exp(\mathcal{L}_i)$. In other words, the composition of two embeddable operations is in general not embeddable, but one can reproduce an arbitrary Markovian operation, according to Def.  \ref{def_MTO}, by composing embeddable ones, analogously to the case of infinitesimal divisible maps and their approximation by the composition of Lindbladian generators \cite{wolf2008dividing}.\\ An equivalent classification can be imposed on the Gibbs-stochastic matrices associated to thermal operations, by considering classical rate matrices $Q$ instead of Lindblad generators. In particular a Gibbs-stochastic matrix $G$ will be called \textit{embeddable} if it can be written as $G=e^Q$, and \textit{Markovian} if it can be approximated arbitrarily well by a sequence of embeddable Gibbs-stochastic matrices, defining the sets $\mathsf{EGS}_n$ and $\mathsf{MGS}_n$, respectively. Indeed, the former definition is related to the well-established \textit{embedding problem} in the theory of stochastic matrices \cite{johansen1974some,davies2010embeddable,baake2020notes}, as other authors \cite{wolf2008assessing,lostaglio2018elementary,lostaglio2021continuous,aguilar2020thermal} have also pointed out.

Now, at least for diagonal states, there is nothing more to thermal operations than their action on populations, i.e., their Gibbs-stochastic matrix, and thus a thermal operation will be (time-independent) Markovian if and only if its corresponding Gibbs-stochastic matrix is so. 
In the following section, we will compute the optimal yield that is achievable under $\mathsf{MTO}_3$ while we defer the computation of the maximum yield under the class $\mathsf{ETO}_3$ to Appendix \ref{appendix_opt}. We will show that memory is ultimately required to achieve the optimal yield $\gamma^*$.\\

{\bf\emph{Markovianity and photoisomerization --}}
Given the considered concept of Markovianity, we can formulate the problem of determining the optimal isomerization yield in terms of the Gibbs-stochastic matrices associated to the thermal operations, and ask whether a certain matrix in $\mathsf{GS}_3$ can be thought of as a snapshot of an underlying classical time-inhomogeneous Markov process.
In other words, we are asking whether the optimal yield $\gamma^*$ generated by the action of $\mathsf{GS}_3$ is strictly larger than the optimal yield $\gamma_{\mathcal{M}}$ generated by the action of $\mathsf{MGS}_3$, and whether the latter is strictly larger than the optimal yield $\gamma_\mathcal{E}$ generated by the action of $\mathsf{EGS}_3$ alone.

\noindent For that, we use the concept of \textit{continuous thermomajorization} recently introduced in \cite{lostaglio2021continuous}.
\begin{definition}[Continuous thermomajorization]
A state $\rho$ is said to continuously thermomajorize a state $\sigma$, denoted $\rho \cthermomaj \sigma$, if there exists a continuous family of states $\{r(s)\mid0\leq s \leq t\}$, with $r(0)=\rho$ and $r(t)=\sigma$, such that $\forall t'>t$ one has $r(t)\thermomaj r(t')$.
\label{def_cthermomaj}
\end{definition}
\noindent
In particular this yields
\begin{theorem}
Consider two quasiclassical states $\rho,\sigma$. Then, one has $\rho\xrightarrow[]{\mathsf{MTO}}\sigma$ if and only if $\rho\cthermomaj\sigma$.
\end{theorem}
\noindent
which is included in Theorem 1 in \cite{lostaglio2021continuous}. Using the theorem above, the optimal Markovian yield can then be computed analytically by making use of some properties of thermomajorization curves and two-levels full thermalizations \cite{perry2018sufficient}. Here we only present the statement of the theorem, while the (rather technical) proof will be presented in Appendix \ref{appendix_markov_gap}.

\begin{theorem}[Optimal Markovian yield]
The optimal yield $\gamma^*$ is not achievable under Markovian thermal operations. Furthermore, the optimal Markovian yield is given by
\begin{equation}
    \gamma_\mathcal{M} = \begin{cases}\Big[ q + (1-q)\frac{e^{-\Delta}}{1+e^{-\Delta}}\Big]\frac{e^{-\Delta}}{e^{-\Delta}+e^{-W}} & \text{if } q\geq\tilde{q}\,, \\[1ex] \Big[ 1 - q\frac{e^{-W}}{e^{-W}+e^{-\Delta}}\Big]\frac{e^{-\Delta}}{1+e^{-\Delta}} & \text{if } q<\tilde{q}\,.\end{cases}
\end{equation}
\label{thm_markovyield}
\end{theorem}
\begin{proof}
See Appendix \ref{appendix_markov_gap}. 
\end{proof}

\begin{figure}[t]
\centering
\includegraphics[]{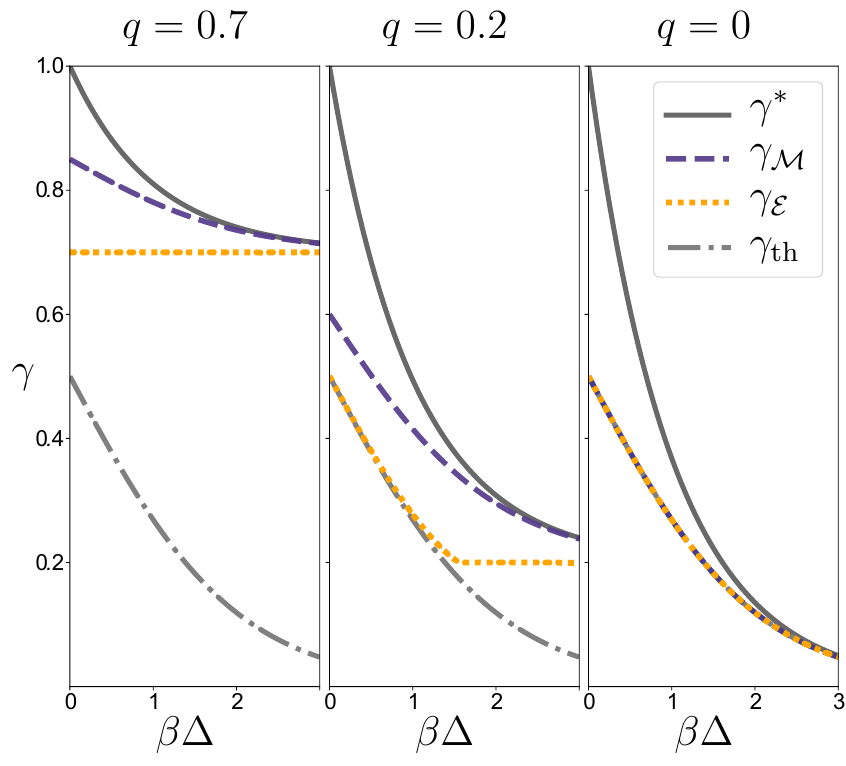}
\caption{Comparison between the optimal yield $\gamma^*$ achievable under general thermal operations, the optimal Markovian yield $\gamma_\mathcal{M}$ and the optimal embeddable yield $\gamma_\mathcal{E}$ obtained via numerical optimization, in the limit $\beta W\to\infty$. The results are shown for three different values of the initial population $q$ (see Eq.\ref{eq_state_param}), and they confirm the inaccessibility of the optimal yield $\gamma^*$ in the absence of memory. The equilibrium value for the yield is denoted by $\gamma_{\rm th}$. Note that in the considered abstract formalism, linking the role of $\beta$ to specific microscopic models is outside the scope of the analysis.}
    \label{fig:Markov_yields}
\end{figure}

The two maximum yields corresponding to \textit{Markovian} and \textit{embeddable} operations (according to Defs. \ref{def_ETO} and \ref{def_MTO}), are represented in Fig.\ref{fig:Markov_yields}, for three different values of the parameter $q$. The optimal embeddable yield obtained numerically is well described by  $\gamma_\mathcal{E} \approx \max\{q,\gamma_{\rm th}\}$, where $\gamma_{\rm th}=e^{-\Delta}/Z$ is the equilibrium value for the yield,
in agreement with the fact that it is always possible to write down a time-independent Lindblad master equation producing a yield $\gamma=q$, as well as $\gamma=\gamma_{\rm th}$.

It is interesting to notice that, when $q=0$, the maximum Markovian yield coincides with the equilibrium population of the second level, i.e. $\gamma_{\rm th}=e^{-\Delta}/Z$. This means that, when starting from an unexcited state, Markovian environments cannot do any better (in the task of optimizing the yield) than just thermalizing the state. One could then reach higher values of the yield by means of some non-Markovian interaction.\\

{\bf\emph{Discussion --}} These results rigorously demonstrate that environmental memory effects 
can significantly increase the efficiency of model bio-molecular quantum processes. 
Notably, the optimal photoisomerization yield is dramatically affected 
by the degree of memory of the environment of the photoisomer. In the
presence of a memoryless environment, which leads to a Markovian system dynamics, the photoisomerization yield is severely limited whenever $\Delta\approx \beta^{-1}$. 
With growing \textit{cis-trans} gap $\Delta$, the advantage offered by non-Markovianity decreases, such
that memory effects do not significantly alter the photoisomerization yield for $\Delta\gg\beta^{-1}$, up until the point in which $\Delta$ grows large enough to be comparable to $W$. Indeed, when $\Delta\approx W$, the optimal yield approaches a value $q$, while the Markovian yield $\gamma_\mathcal{M}$ comes very close to $q/2$. Thus, the effect of memory on the photoisomerization yield is maximal precisely when $\Delta\approx \beta^{-1}$ or $\Delta\approx W$. On the other hand, when $\Delta$ is far from both $0$ and $W$, the boost offered by memory is reduced.
These results also offer a quantitative comparison between the power of time-inhomogeneous Markovian 
processes versus their time-homogeneous counterparts. In particular, if one relaxes embeddability of 
the channels to CP-divisibility, the achievable yields grow significantly, but still not 
enough to saturate the bound $\gamma^*$, which can only be achieved by means of non-Markovian processes.
We stress that these results go beyond specific choices of environments and dynamics.
Our approach has the additional merit of providing an athermality witness in experiments. Indeed, a violation of the bound $\gamma\leq\gamma^*$, would necessarily imply that the process at hand is not a thermal operation, and therefore the environment is not merely acting as a passive thermal bath, but provides the system with athermality resources. In other words, measuring a violation of the bound would imply the presence of external batteries and/or clocks coupled to the molecule. 
On the other hand, under the assumptions of Markovianity, a violation of the bound $\gamma\leq\gamma_{\mathcal{M}}$ would allow to rigorously discard any microscopic model that is locally phase-covariant (e.g. any intermediate dynamics which does not mix populations and coherences).
Indeed, the current analysis assumes time-translation covariance of the generators of
the dynamics which in turn assumes the absence of an external clock \cite{horodecki2013fundamental,woods2019autonomous}, the addition of which would carry a thermodynamical cost, but would allow the global 
thermal operation to be underpinned by processes that are not phase covariant at every step in the evolution. The resulting 
thermal operation could however still obey an extended definition of Markovianity, in the sense of being underpinned by a divisible map composed of non-phase covariant intermediate steps \cite{Haase2019NPC}. Operations of this
type cannot be analyzed within the current approach due to the resulting mixing of populations and coherences, and it is an open problem to determine the optimal yield of this class of clock-assisted operations. 
Our results have been made possible by the mathematical rigour of the resource theoretical
approach to thermodynamics and suggest that other features such as the role of
spatial correlations in the environment can be treated in an analogous fashion, thus establishing a theoretical framework to quantitatively assess the role of temporal and spatial correlation in assisting the efficiency of dynamical processes at the nanoscale.\\

{\em Acknowledgements:} This work was supported by the European Research Council Synergy Grant HyperQ (Grant no. 856432). 
We thank Matteo Lostaglio, Kamil Korzekwa and Nicole Yunger-Halpern for their comments on the initial version of this manuscript.

\bibliographystyle{unsrt}
\bibliography{ms.bib}


\onecolumngrid
\appendix
\section{Comparison with a $4$-dimensional model}
\label{appendix_GS4}

The model used in \cite{halpern2020fundamental}, reduces a molecular switch to a four levels system. On the other hand, we have claimed that one obtains equivalent results if the Hilbert space of the system gets truncated down to a three-dimensional one. We know want to give more rigorous ground to such claim. Let us consider the full $4$-dimensional case, i.e. let $\mathsf{GS}_4$ be the set of $4\times 4$ Gibbs-stochastic matrices. If we parametrize $G\in \mathsf{GS}_4$ as

\begin{equation}
    G=\begin{pmatrix} 1-g_1e^{-\Delta}-g_2e^{-W}-g_5 e^{-W'} & g_1 & g_2 & g_5 \\ (1-g_3)e^{-\Delta} - g_4 e^{-W} - g_6 e^{-W'} & g_3 & g_4 & g_6 \\ (1-g_8) e^{-W} -g_7 e^{-\Delta} -g_9 e^{-W'} & g_7 & g_8 & g_9 \\ (g_5{+}g_6{+}g_9)e^{-W'}{-}(1{-}g_1{-}g_3{-}g_7)e^{-\Delta}{-}(1{-}g_2{-}g_4{-}g_8)e^{-W} & 1{-}g_1{-}g_3{-}g_7 & 1{-}g_2{-}g_4{-}g_8 & 1{-}g_5{-}g_6{-}g_9 \end{pmatrix}\,,
    \label{eq:generic_GS4}
\end{equation}
we see that we have now 9 parameters $g_1,\dots,g_9\in[0,1]$, further constrained in such a way to ensure the non-negativity of each entry. Clearly we regain the results that we already found in three dimensions, by setting $g_7=1-g_1-g_3$, $g_8=1-g_2-g_4$ and $g_5=g_6=g_9=0$ as such a choice allows the decomposition $G=G^{(3)}\oplus\mathbb{1}_1$, where $G^{(3)}\in \mathsf{GS}_3$, and $\mathbb{1}_1$ is the identity operator in $1$ dimension. This implies that the optimal yield in $\mathsf{GS}_4$ is lower-bounded by the optimal yield in $\mathsf{GS}_3$. In other words,
\begin{equation}
    \sup_{G\in \mathsf{GS}_4} \gamma(G) \geq \sup_{G\in \mathsf{GS}_3} \gamma(G)\,.
    \label{eq:gs4_gs3_lb}
\end{equation}
By assumption, the fourth level is never populated initially, i.e. the initial state has the form $(1-q,0,q,0)$, while after the action of $G$ the final state has a yield
\begin{equation}
    \gamma^{(4)}(g_3,g_4,g_6)= (1-q)\Big[(1-g_3)e^{-\Delta} \Big] + \Big[ q - (1-q)e^{-W}\Big] g_4 - (1-q)e^{-W'}g_6 \,.
\end{equation}
By recalling the analogous expression for the yield in the three-dimensional case (and denoting it with the symbol $\gamma^{(3)}(g_3,g_4)$), we can easily see that
\begin{equation}
    \gamma^{(4)}(g_3,g_4,g_6) \equiv \gamma^{(3)}(g_3,g_4)- (1-q)e^{-W'}g_6 \,,
\end{equation}
which is clearly less than or equal to $\gamma^{(3)}(g_3,g_4)$, for any value of $g_6$. In other words, we have the inequality
\begin{equation}
    \sup_{G\in \mathsf{GS}_4} \gamma(G) \leq \sup_{G\in \mathsf{GS}_3} \gamma(G)\,.
\end{equation}
This, together with Eq.\eqref{eq:gs4_gs3_lb}, proves that the optimal photoisomerization yield can always be optimized in $\mathsf{GS}_3$, without the need of a fourth level.

\section{Phase covariance and Markovian yield}
\label{appendix_quantum_advantage}

We have proven that no $\mathsf{MGS}_3$ process can achieve yields larger than $\gamma_\mathcal{M}$. We know that all thermal operations are in our case identified with their action on populations, i.e. they are essentially classical. However, one could then naturally ask whether a Markovian thermal operation would be able to achieve yields that classically are only possible through memory effects. If one considered general quantum operations, the answer would be yes, and in particular, as showed in \cite{korzekwa2021PRX}, it is possible to simulate classical memory via quantum Markovian operations. What this tells us, is that there exist some quantum operations that are memoryless, but induce non-embeddable Gibbs-stochastic processes on population vectors. However, this is only possible because quantum operation can in principle manipulate coherences and convert them into populations, an effect that can induce a backflow that classically would require memory (as discussed in \cite{korzekwa2021PRX}). It is then natural to ask if this still holds when we impose the additional constraint of time-translation covariance, which prevents us from accessing coherence without an external quantum clock.\\ Indeed, such an additional assumption guarantees that this effect cannot take place, and the only way to achieve yields larger then $\gamma_\mathcal{M}$ through thermal operations is to exploit memory. This is precisely the content of the following theorem.
\begin{theorem}
Any Markovian (embeddable) time-translation covariant quantum channel that preserves the Gibbs state of a non-degenerate Hamiltonian, induces a Markovian (embeddable) Gibbs-stochastic mapping on the populations in the energy eigenbasis.
\end{theorem}
\begin{proof}
The proof starts by first showing that under a time-translation covariant evolution, the diagonal elements of the density operator decouple from its off-diagonal elements. Then it proceeds by showing that under a channel that is Markovian, defined as a possessing a description as a Lindblad equation with (possibly) time-dependent coefficients, the populations obey a classical rate equation.\\
Let us consider a time-translation covariant quantum channel $\mathcal{E}_t$, an initial state $\rho(0)$ and the resulting trajectory $\rho(t)=\mathcal{E}_t(\rho)$. We denote with $p(t)$ the vector of diagonal elements of $\rho(t)$, and write $\rho^{(0)}(t) = \text{diag}(\rho(t))$. Under time-translation covariance of $\mathcal{E}_t$, which implies $ \mathcal{E}_t^{(0)}(\rho) = \mathcal{E}_t(\rho^{(0)}) $, we find
\begin{equation} \rho^{(0)}(t+dt) - \rho^{(0)}(t) = \mathcal{E}_{t+dt}^{(0)}(\rho)-\mathcal{E}_t^{(0)}(\rho) = \mathcal{E}_{t+dt}(\rho^{(0)})-\mathcal{E}_t(\rho^{(0)}) = \Big[\mathcal{E}_{t+dt}-\mathcal{E}_{t}\Big](\rho^{(0)})\,,\end{equation}
which, under the assumption of Markovianity becomes
\begin{equation} \dot{\rho}^{(0)}(t) = \mathcal{L}_t \rho^{(0)}(t)\,.\end{equation}
By writing this equation in the energy eigenbasis $\{\ket{i}\bra{j}\}$, and exploiting the fact that the Hamiltonian is non-degenerate, we find
\begin{equation} \sum_j \dot{p}_j(t)\ket{j}\bra{j} = \sum_j p_j(t)\mathcal{L}_t(\ket{j}\bra{j})\,, \end{equation}
whose evolution equation for the $i$-th diagonal element is given by
\begin{equation} \dot{p}_i(t) = \sum_j \bra{i}\mathcal{L}_t(\ket{j}\bra{j})\ket{i} p_j(t) = \sum_j L_{ij}(t)  p_j(t)\,,\end{equation}
which implies the classical rate equation
\begin{equation}\dot{p}(t)=L(t)p(t)\,,\end{equation}
with the generator
\begin{equation} L_{ij}(t)= \bra{i}\mathcal{L}_t(\ket{j}\bra{j})\ket{i} \end{equation}
and solution
\begin{equation} p(t)=\mathsf{T}\exp\Big(\int_0^t L(\tau)d\tau\Big)\,p(0)\,. \end{equation}
Thus, the populations follow a classical Markovian process and the Gibbs-stochastic process associated to it is therefore Markovian by construction.
\end{proof}

The result above can be applied to the set of thermal operations, and in particular it guarantees that memoryless thermal operations always induce memoryless processes on population vectors, as shown in Fig. \ref{fig:MTO_EGS}.

\begin{figure}[h]
    \centering
    \includegraphics[]{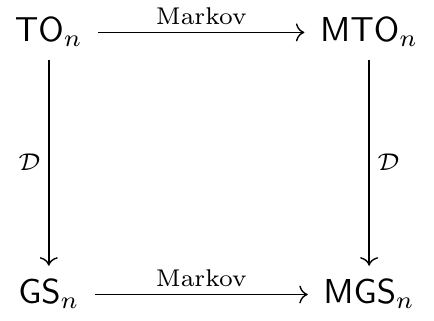}
    \caption{If $\mathcal{D}$ is a mapping that takes thermal operations to their corresponding Gibbs-stochastic processes, time-translation covariance guarantees that $\mathcal{D}$ also maps memoryless thermal operations to memoryless Gibbs-stochastic processes.}
    \label{fig:MTO_EGS}
\end{figure}

This then proves that $\mathsf{MTO}_3$ operations are identified with $\mathsf{MGS}_3$ processes on population vectors, and therefore they can at most produce a yield $\gamma_\mathcal{M}$.

\section{Optimal Markovian yield}
\label{appendix_markov_gap}
\noindent We present here the proof to Theorem \ref{thm_markovyield}.

\begin{theorem}[Optimal Markovian yield]
The optimal yield $\gamma^*$ is not achievable under Markovian thermal operations. Furthermore, the optimal Markovian yield is given by
\begin{equation}
    \gamma_\mathcal{M} = \begin{cases}\Big[ q + (1-q)\frac{e^{-\Delta}}{1+e^{-\Delta}}\Big]\frac{e^{-\Delta}}{e^{-\Delta}+e^{-W}} & \text{if } q\geq\tilde{q}\,, \\[1ex] \Big[ 1 - q\frac{e^{-W}}{e^{-W}+e^{-\Delta}}\Big]\frac{e^{-\Delta}}{1+e^{-\Delta}} & \text{if } q<\tilde{q}\,.\end{cases}
\end{equation}
\label{thm_optmarkov}
\end{theorem}
\begin{proof}
In this proof we will make use of some properties of continuous thermomajorization that can be found in \cite{lostaglio2021continuous}, in particular we will exploit the behaviour of thermomajorization curves under two-levels full thermalizations.\\
A necessary condition for a yield $\gamma$ to be achievable by $\mathsf{MTO}$ is the existence of a final state $(p_0,\gamma,p_W)$ such that \begin{equation}(1-q,0,q)\cthermomaj(p_0,\gamma,p_W)\end{equation}
But this implies that there is a set of at most $3!-1=5$ full thermalizations $G_{ij}$, between adjacent points $i,j$ on the thermomajorization curve $f(x)$ of the initial state, such that their composition brings the initial state $(1-q,0,q)$ to one that has the same $\beta$-ordering of the final state and also thermomajorizes it. We denote the possible $\beta$-orderings by permutations of the triplet $(0,\Delta,W)$. Now, the maximum yield $\gamma^*$ can only be achieved if the $\beta$-ordering of the final state is of the form $(\Delta,x,y)$, i.e. $\gamma^*$ is the largest $\beta$-ordered probability. This also implies that $\gamma^*=f(e^{-\Delta})$. On the other hand, the initial $\beta$-ordering is of the form $(x',y',\Delta)$. Since full thermalizations have to be between adjacent points, and since in a triplet $(a,b,c)$ there are only two ways to choose adjacent points (namely $(a,b)$ and $(b,c)$), they necessarily have to alternate (picking the same pair twice in a row has no effect since full thermalizations are idempodent). Therefore, in order to transform the initial $\beta$-ordering into one that has $\Delta$ as first element, we only have the two following possibilities, differing only from which pair is chosen for the first thermalization:\\
\textbf{A)} $(x',y',\Delta)\mapsto(y',x',\Delta)\mapsto(y',\Delta,x')\mapsto(\Delta,y',x')$\\
\textbf{B)} $(x',y',\Delta)\mapsto(x',\Delta,y')\mapsto(\Delta,x',y')$\\
We can conclude that, in order to obtain the $\beta$-ordering of the final state, at least two full thermalizations are needed: one between levels $y'$ and $\Delta$, and the other between levels $x'$ and $\Delta$. This means however that all segments forming $f(x)$ will be lowered by a finite amount, producing a new curve $\tilde{f}(x)$ such that $\tilde{f}(x)<f(x)$ $\forall x\in(0,Z)$ (see Figure \ref{fig:gap_proof} for visual reference). In particular, after the two full thermalizations, no curve $g(x)$ that is below $\tilde{f}(x)$ can have $g(e^{-\Delta})=\gamma^*$. Indeed, $\gamma_\mathcal{M}\leq \tilde{f}(e^{-\Delta})< f(e^{-\Delta})$, and therefore there is a finite difference between $\gamma_\mathcal{M}$ and $\gamma^*$.\\
We now find an exact expression for the optimal Markovian yield, and we focus on the case $q\geq\tilde{q}$ only, which is the one that is physically more relevant. The case $q<\tilde{q}$ is analogous.\\
If we want to find the optimal Markovian yield exactly, we need to check both paths \textbf{A)} and \textbf{B)} above, and select the one that gives the highest yield. It turns out that the best path is the one with only two thermalizations, i.e. path \textbf{B)}. We start from the thermomajorization curve of the initial state (for the ordering $(W,0,\Delta)$), which has the form
\begin{equation}
f(x)=\begin{cases} q e^W x & x<e^{-W}\\ q + (1-q)(x-e^{-W}) & e^{-W} \leq x \leq 1+e^{-W} \\ 1 & 1+e^{-W}< x \leq Z\,,
\end{cases}
\end{equation}
while the two full thermalizations of path B) produce the curve
\begin{equation}
\tilde{f}(x)=\begin{cases} \Big[ q + (1-q)\frac{e^{-\Delta}}{1+e^{-\Delta}}\Big]\frac{x}{e^{-\Delta}+e^{-W}} & x\leq e^{-\Delta}+e^{-W}\\[1ex]  q + (1-q)\frac{x-e^{-W}}{1+e^{-\Delta}} & e^{-\Delta}+e^{-W}< x \leq Z
\end{cases}
\end{equation}
with ordering $(\Delta,W,0)$. Now, to obtain the exact $\beta$-ordering of the final state, a third full thermalization might be needed to swap $W$ and $0$, but this would only affect the second segment, i.e. $x>e^{-\Delta}+e^{-W}$, while the optimal Markovian yield $\gamma_\mathcal{M}=\tilde{f}(e^{-\Delta})$ is computed on the first, and is therefore unaffected.
We conclude that the optimal Markovian yield is then
\begin{equation}
    \gamma_\mathcal{M} = \Big[ q + (1-q)\frac{e^{-\Delta}}{1+e^{-\Delta}}\Big]\frac{e^{-\Delta}}{e^{-\Delta}+e^{-W}}\,.
\end{equation}
\end{proof}

\begin{figure}[h]
\centering
\includegraphics[]{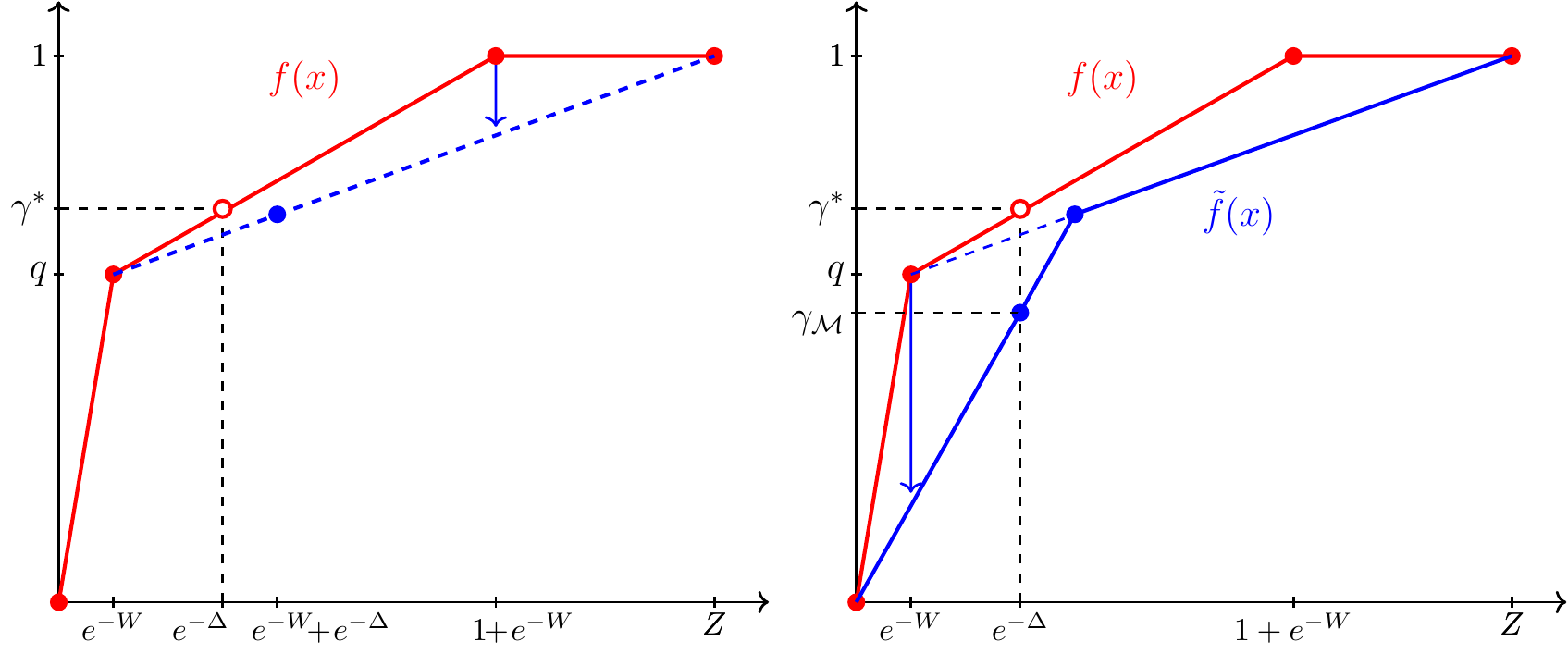}

\caption{Thermomajorization curves for the proof of Theorem \ref{thm_optmarkov}. (LEFT): The red solid curve $f(x)$ is the thermomajorization curve of the initial state $(1-q,0,q)$, with $q\geq\tilde{q}$, i.e. corresponding to the $\beta$-ordering $(W,0,\Delta)$. The maximum allowed yield is $\gamma^*=f(e^{-\Delta})$. The dashed blue curve is the result of the first full thermalization between levels $\Delta$ and $0$. The second and third segments are brought down, and the new elbow point is at $x=e^{-\Delta}+e^{-W}$. (RIGHT): After the second full thermalization the first and the second segments are brought down and the curve $\tilde{f}(x)$ is obtained. It is clear that $\gamma^*$ is no longer accessible, if the final state has to be thermomajorized by $\tilde{f}(x)$.}
\label{fig:gap_proof}
\end{figure}

\section{Optimal embeddable yield}
\label{appendix_opt}

\noindent The problem of assessing whether a $n\times n$ stochastic matrix is embeddable is considered to be fully solved for $n\leq 3$ only. Indeed, for the case $n=3$, by merging two results by Johansen \cite{johansen1974some} and Carette \cite{carette1995characterizations}, it is possible to give necessary and sufficient conditions for the embeddability of $3\times3$ Gibbs-stochastic matrices. However, since we will only be interested in the simpler case of matrices with real spectra, we hereby only state a simpler, partial characterization. If we denote the set of embeddable $n\times n$ Gibbs-stochastic matrices by $\mathsf{EGS}_n$, we have the following theorem.
\begin{theorem}[Partial Characterization of $\mathsf{EGS}_3$] Consider a matrix $G\in \mathsf{GS}_3$ with real spectrum $\sigma(G)=\{1,\lambda_1,\lambda_2\}$. Then the following statements are true: \\
\textbf{a)} If $0\in\sigma(G)$, then $G$ is not embeddable.\\ \textbf{b)} If either $\lambda_1<0$ or $\lambda_2<0$, then $G$ is embeddable only if $\lambda_1=\lambda_2$.\\ \textbf{c)} If $\lambda_1,\lambda_2>0$\, then $G$ is embeddable if and only if
\[ G_{ij} \geq  f(\lambda_1,\lambda_2)\frac{e^{-\beta E_i}}{Z} \quad \forall\, i\neq j  \,,\]
with $ f =  \frac{(\lambda_2-1)\ln\lambda_1 - (\lambda_1-1)\ln\lambda_2}{\ln\lambda_2 - \ln\lambda_1}\,. $
\label{thm:gsembed}
\end{theorem}
\begin{proof}
\textbf{a)} If a matrix is embeddable, it cannot have vanishing eigenvalues, as the spectral mapping theorem forces them to be of the form $\lambda=e^{\theta}$, which is never zero. \\
\textbf{b)} Every negative eigenvalue $\lambda$ of an embeddable matrix must have even algebraic multiplicity (see Lemma 2 in \cite{davies2010embeddable} for a proof). When $n=3$, this means that if one of the two eigenvalues $\lambda_1,\lambda_2$ is negative, then the other one is also negative, and equal to the first.\\
\textbf{c)} This follows directly from the application of Johansen's result (in particular, it follows from Corollary 1.2 in \cite{johansen1974some}), and from imposing the additional constraint of Gibbs-stochasticity. The domain of $f$ is extended to $\lambda_1=\lambda_2$ by continuity.
\end{proof}

\noindent To simplify the problem, since the typical energy scales in photoisomers are such that $e^{-\beta W} \ll 1$ (with $\beta W \approx 90$ at room temperature for azobenzene), we perform the limit $W\to\infty$. In this limit, a generic Gibbs-stochastic matrix takes the form
\begin{equation}
    \tilde{G} = \begin{pmatrix} 1-g_1e^{-\Delta} & g_1 & g_2 \\ g_1e^{-\Delta} & 1-g_1 & g_4 \\ 0 & 0 & 1-g_2-g_4 \end{pmatrix}\,,
    \label{eq:generic_GS3_inf}
\end{equation}
with $g_2+g_4\leq 1$ being the only remaining constraint.

The spectrum of $\tilde{G}$ is then easily obtained as:
\begin{equation}
\begin{split}
      &\sigma(\tilde{G}) = \Big\{1\,,\,\,1-g_1(1+e^{-\Delta})\,,\,\,1-g_2-g_4\Big\}\,.  
\end{split}
\end{equation}

Theorem \ref{thm:gsembed} then provides a full characterization of embeddability for $\tilde{G}$, which results in tighter constraints on the parameters $g_1,g_2,g_4$.\\
The problem of optimizing the embeddable yield (i.e. the yield achievable under Markovian processes with time-independent generators), in the limit $W\to\infty$, can be formulated as
\begin{equation}
\begin{aligned}
\max_{g} \qquad & \gamma(g_1,g_4)= (1-q)e^{-\Delta}g_1+q g_4\\
\textrm{s.t.} \qquad &  f(g_1,g_2,g_4)\leq (1+e^{-\Delta})g_1 \leq 1\\
  & (1+e^{-\Delta})g_2 \geq f(g_1,g_2,g_4)    \\
  & (1+e^{-\Delta})g_4 \geq f(g_1,g_2,g_4)e^{-\Delta}    \\
  & g_2+g_4 \leq 1    \\
\end{aligned}
\end{equation}
where
\begin{equation}
    f(g_1,g_2,g_4) = \frac{g_1(1+e^{-\Delta})\ln(1-g_2-g_4)-(g_2+g_4)\ln(1-g_1(1+e^{-\Delta}))}{\ln(1-g_2-g_4)-\ln(1-g_1(1+e^{-\Delta}))}\,.
\end{equation}
Despite the objective function being linear, the problem is extremely hard to handle, given the highly non-linear nature of the constraints. However, we can simplify its form by performing the change of variables
\begin{equation}
    \begin{cases}
    k_1 = (1+e^{-\Delta})g_1\,, \quad & 0\leq k_1 < 1\,, \\
    k_2 = g_2+g_4 \,, \quad  & 0\leq k_2 < 1\,,\\
    k_3 = g_2-g_4 \,, \quad  & -1\leq k_3 \leq 1\,,
    \end{cases}
\end{equation}
in terms of which the problem becomes
\begin{equation}
\begin{aligned}
\max_{k} \qquad & \gamma(k_1,k_2,k_3)= (1-q)\gamma_{\rm th}k_1+\frac{1}{2}q (k_2 - k_3 ) \\
\textrm{s.t.} \qquad & k_1\geq f(k_1,k_2)\\
  & k_2+k_3 \geq 2f(k_1,k_2)/Z  \\
  & k_2-k_3 \geq 2f(k_1,k_2) e^{-\Delta}/Z   \\
\end{aligned}
\end{equation}
where
\begin{equation}
    f(k_1,k_2) = \frac{k_1\ln(1-k_2)-k_2\ln(1-k_1)}{\ln(1-k_2)-\ln(1-k_1)}\,.
\end{equation}
It is now easier to see that the first constraint is trivially satisfied for all values of $k_1$ and $k_2$, leaving us with the following 
\begin{lemma}
The optimal photoisomerization yield achievable under $\mathsf{EGS}_3$ operations, in the limit $\beta W\to\infty$, is the solution to the optimization problem 
\begin{equation}
\begin{aligned}
\min_{k} \qquad & -\gamma(k_1,k_2,k_3)\\
\textrm{s.t.} \qquad   & 2f(k_1,k_2)/Z - k_2 - k_3 \leq 0  \\
  & 2f(k_1,k_2)e^{-\Delta}/Z - k_2 + k_3 \leq 0
\end{aligned}
\end{equation}
where $k_1\in[0,1)$, $k_2\in[0,1)$, $k_3\in[-1,1]$, and $\gamma(k_1,k_2,k_3)=(1-q)\gamma_{\rm th}k_1+\frac{1}{2}q (k_2 - k_3 )$.
\end{lemma}
As already discussed, the non-linear nature of the constraints prevents the problem from being solved exactly, and therefore the solution provided in the main text is obtained by numerical optimization.\\

\end{document}